\documentclass[a4paper,12pt]{article}
\usepackage[utf8]{inputenc}
\usepackage[english]{babel}
\usepackage{amsmath, amsfonts, amssymb, amsthm}
\usepackage{multirow}
\usepackage{graphicx}
\usepackage{fancyhdr}
\usepackage{hyperref}
\usepackage{float}
\usepackage{url}
\usepackage{pdfpages}

\usepackage[top=1in, bottom=1.25in, left=1in, right=1 in]{geometry}
\usepackage{amssymb, amsmath}
\usepackage{bbm} 
\usepackage{xcolor}
\usepackage{graphicx}
\usepackage{adjustbox}

\newcommand{\email}[1]{\href{mailto:#1}{#1}}

\newcommand{\pr}[1]{(#1(t), \, 0\leq t\leq 1)}
\newcommand{\prT}[1]{(#1_t, \, 0 \leq{t} \leq{T})}
\newcommand{\prTmenos}[1]{(#1_t, \, 0 \leq{t} <{T})}
\providecommand{\keywords}[1]
{
  \small	
  \textbf{\textit{Keywords---}} #1
}

\ifpdf
  \DeclareGraphicsExtensions{.eps,.pdf,.png,.jpg}
\else
  \DeclareGraphicsExtensions{.eps}
\fi

\newcommand{\cA}{\mathcal{A}}

\newcommand{\cF}{\mathcal{F}}
\newcommand{\cG}{\mathcal{G}}
\newcommand{\cH}{\mathcal{H}}

\newcommand{\cS}{\mathcal{S}}
\newcommand{\cT}{\mathcal{T}}
\newcommand{\cX}{\mathcal{X}}

\newcommand{\bF}{\mathbb{F}}
\newcommand{\bG}{\mathbb{G}}

\newcommand{\bOne}{\mathbbm{1}}
\newcommand{\bR}{\mathbb{R}}
\newcommand{\bT}{\mathbb{T}}

\newcommand{\PP}{\boldsymbol{P}}
\newcommand{\QQ}{\boldsymbol{Q}}
\newcommand{\EE}{\boldsymbol{E}}

\providecommand{\keywords}[1]
{
  \small	
  \textbf{\textit{Keywords---}} #1
}

\ifpdf
  \DeclareGraphicsExtensions{.eps,.pdf,.png,.jpg}
\else
  \DeclareGraphicsExtensions{.eps}
\fi

\title{Insider information and its relation with the arbitrage condition and the utility maximization problem }

\author{Bernardo D'Auria
\thanks{This research was partially supported by the Spanish Ministry of Economy and Competitiveness Grants MTM2017-85618-P via FEDER funds and MTM2015-72907-EXP; 
both authors thank the NYUAD, Abu Dhabi (United Arab Emirates) for hosting them during the fall 2018.
}
\thanks{UC3M, Department of Statistics, 28911, Legan\'{e}s, Spain \& UC3M-BS, Institute of Financial Big Data, 28903 Getafe, Spain 
  (\email{bernardo.dauria@uc3m.es}).}
\and Jos\'{e} Antonio Salmer\'{o}n
\thanks{This author acknowledges financial support by an FPU Grant (FPU18/01101) of Spanish Ministerio de Ciencia, Innovaci\'{o}n y Universidades.}
\thanks{UC3M, Department of Statistics, 28911, Legan\'{e}s, Spain 
  (\email{joseantonio.salmeron@uc3m.es}).}
}

\numberwithin{equation}{section}
\DeclareMathOperator*{\argsup}{arg\,sup}

\usepackage{color}
\usepackage{amsmath,amsfonts,mathtools}
\usepackage{comment}
\usepackage{multirow} 
\usepackage{enumitem} 
\usepackage{standalone} 
\usepackage{etoolbox} 
\usepackage{pgfplots} 
\usepackage{xspace} 
\newtheorem{theorem}{Theorem}[section]

\newtheorem{Lemma}[theorem]{Lemma}
\newtheorem{proposition}[theorem]{Proposition}
\newtheorem{corollary}[theorem]{Corollary}
\newtheorem{definition}[theorem]{Definition}

\newtheorem{remark}[theorem]{Remark}
\providecommand{\norm}[1]{\lVert#1\rVert}
\providecommand{\abs}[1]{\lvert#1\rvert}
\providecommand{\citep}[1]{(\cite{#1})}

\newcommand{\Ito}{It\^{o}\xspace}

\begin{document}

\maketitle
\begin{abstract}
    Within the well-known framework of financial portfolio optimization, we analyze the existing
    relationships between the condition of arbitrage and the utility maximization in presence of \emph{insider information}.
    We assume that, since the initial time, the information flow is altered  by adding the knowledge of an additional random variable including future information.
    In this context we study the utility maximization problem under the logarithmic and the Constant Relative Risk Aversion (CRRA) utilities, with and without the restriction of no temporary-bankruptcy. 

    In particular, we show that the value of the insider information may be bounded while the arbitrage condition holds 
    and 
    we prove that the insider information does not always imply arbitrage for the insider by providing an explicit example.

\end{abstract}

\keywords{
optimal portfolio; 
enlargement of filtration; 
value of the information; 
arbitrage; 
No Free Lunch Vanishing Risk;
risk neutral measure
}

\maketitle

\section{Introduction}
In the field of financial mathematics, the problem of the optimal portfolio plays a crucial role and in recent years it has been deeply analyzed in the literature. In its simplest form it consists in finding the best strategy in order to maximize a given utility function at a fixed terminal finite time.

In the simplified settings of just two assets, one risk-less and one risky, the optimal portfolio problem has been introduced and solved in \cite{merton1969} by considering both the logarithmic and risk adverse utilities.
Later, a more general model was considered in \cite{merton1971}. 

One fundamental ingredient in computing the optimal portfolio is the information flow that the agent employs in order to build her strategy. This flow is modeled mathematically by the concept of filtration, and the restrictions on the agent choices are modeled by requiring that her portfolio is adapted to this filtration.
While in general the underlying filtration is the one naturally generated by the set of risky assets, in the literature the interest has recently grown about analyzing filtrations that contain additional information.

The technique consisting in substituting the natural filtration by a larger one containing additional information, is generally referred to as \emph{enlargement of filtration} and has been introduced and studied in the seminal works \cite{JeulinYor1978,Jacod1985,Yor1985,ImkellerFollmer93}.
The concept of enlargement of filtration was first applied to the financial setting by \cite{karatzas1996}, to describe situations in which the agent has access to privileged information and to model the insider trading portfolios, and in \cite{GrorudPontier1998}, as a first attempt to detect the use of insider information by applying statistical tests.

In \cite{karatzas1996}, various examples of initial enlargements were analyzed by computing the expected additional gain carried by the privileged information. 
It was shown that the knowledge of the price of the stock at a given moment in the future implies an expected unbounded additional profit while the knowledge of an interval of values containing the future price only added a bounded expected additional gain. 
For the case of an infinite interval, a direct proof of this result was given in \cite{karatzas1996} while, for the case of a finite interval, the result was only conjectured by the support of numerical calculations. Later a series of  remarkable works, \cite{AmendingerImkellerSchweizer1998, Amendinger2003}, employed the concept of the \emph{Fischer Entropy Theory}, shortly mentioned already in \cite{karatzas1996}, to close the conjecture. In Theorem \ref{thm:karatzas1996.conjecture} below we prove again this result by the same techniques used in \cite{karatzas1996}. In \cite{AnkirchnerImkeller06}, the Fischer Entropy Theory is generalized for a more broad class of enlargement filtration. 

In more recent years many results on insider trading models appeared, we just mention \cite{Baudoin2004,Biagini2005,Aksamit2017} and references therein.

The above research indicates that there should exist a relation between the type of additional information, such as if it is exact or it is of interval type, and the value that it carries in terms of its contribution to the maximal expected utility. A partial result to this question has been given in \cite{AmendingerImkellerSchweizer1998} where they look at the atomicity of the insider information.
However there are many questions still open, such as understanding the existence of the arbitrage condition is related to the boundedness of the value of the information.

On this direction, \cite{AnkirchnerImkeller05} analyzes the relation between the arbitrage condition -- in particular the (NFLVR), see Proposition \ref{FLVR.and.alpha} below -- and the integrability of the drift of the semi-martingale representation of the asset in the enlarged filtration, see Proposition \ref{drift.decomposition} below.
The works \cite{Fontana16, Fontana18,ChauRunggaldierTankov} studied the relations of a weak arbitrage condition, the No Unbounded Profits with Bounded Risk property (NUPBR), with different types of  enlargements of filtrations.

It seems that an important ingredient in the analysis is the role played by the set of strategies that are allowed to be used.
For example, the strategy constructed in \cite{karatzas1996} to play with the information given by an interval of future prices, does not take advantage of the arbitrage condition introduced by this information.
This is due to the fact that the proposed strategy avoids the possibility for the insider agent to be for some moments in time in bankruptcy, what we later refer to as \emph{temporary-bankruptcy}. 
However it is possible to construct, by removing this constraint, even simpler strategies that get advantage of the arbitrage condition and assure a positive gain, see for example Propositions \ref{prop:strategy.semi-infinite.inteval} and \ref{prop:strategy.finite.inteval} below, for the semi-infinite and the finite interval respectively.

In this work, we analyze the relations that may exist between the condition of arbitrage and the utility maximization problem under the special setting of enlargement of filtration.
We first study the condition of (NFLVR), introduced in \cite{Delbaen1992}, and its relationship with bankruptcy,
by concluding that if the privileged information implies arbitrage, the investor can improve her profit expectations by employing strategies that allow temporary-bankruptcy. 
Even if it may seem counterintuitive, the privileged information guarantees that the trend will be corrected and the condition of bankruptcy will be only temporary. 

Finally we prove the result that the insider information does not always imply arbitrage, by constructing a counter example in Section \ref{sec:example.3}. This is surprising if we consider that for the case of progressive enlargement, it has been proved that often the additional information implies arbitrage, for more information see \cite{Imkeller2002, Fontana2014,Kardaras2014}.

We verify that the introduction of the privileged information does not violate the Novikov condition, Equation \eqref{novikov.condition} below, therefore assuring the existence of an equivalent local martingale measure and therefore the absence of arbitrage.

In Section \ref{sec:basic.notions}, we provide the notation as well as the basic and preliminary notions that we adopt for the rest of the paper.
It includes a more precise definition of the general framework in which the problem of the optimal portfolio is framed,
such as the definition of arbitrage and the concept of enlargements of filtration.
In Section \ref{sec:utility.max.prob}, we introduce the utility maximization problem, analyzing it under different conditions on the set of allowed strategies, such as the no-temporary-bankruptcy, and linking it with the arbitrage conditions.
Section \ref{sec:examples} provides three examples that show various cases of enlargement of filtration, where the arbitrage condition is analyzed and  the utility maximization problem is solved.
Section \ref{sec:conclusions} ends with some conclusions.

\section{Basic notions and preliminaries}\label{sec:basic.notions}
As a general setup we assume to work in a probability space
$(\Omega, \cF,\bF,\PP)$
where $\cF$ is the event sigma-algebra, and $\bF=\{\cF_t, t\geq0\}$ is a right continuous filtration satisfying the usual conditions.
We consider a financial market with a continuous $\mathbb{R}^2$ semi-martingale $\cS = (D,S)$ 
and, unless otherwise specified, the filtration $\bF$ that is the natural one generated by $\cS$. 
Each component represents the prices of an asset in which the agent could invest. 
Usually, the first one is assumed to be a risk-less asset driven by some interest rate $r>0$, and the second one is given by a diffusion process whose coefficients $\mu$ and $\sigma$ are $\bF$-adapted processes.
In this paper, we generally assume that the dynamics of the semi-martingale $\cS$ are given by the following stochastic differential equations,
\begin{subequations}\label{def.asset}
    \begin{align}
     dD_t &= D_t r \, dt \label{asset.riskless} \\
     dS_t &= \mu(t,S_t) \, dt + \sigma(t,S_t) \, dB_t  \label{asset.risky}
    \end{align}
\end{subequations}
with $D_0=1$ and where the process $B = \prT{B}$ is a standard $\bF$-adapted Brownian motion.
We fix $T>0$ as a horizon time, assumed to be finite.
The semi-martingale assumption implies that the stochastic integral operator is well defined. 
So, for an $\mathbb{R}^2$ valued predictable and $\bF$-adapted process $H=(H^1,H^2)$, 
we use the notation $(H\cdot\cS)_t$ to denote the sum of the components of the \Ito integral,
i.e. 
\begin{align}
    (H\cdot\cS)_t &= \int_0^t H^1_udD_u+\int_0^t H^2_udS_u \quad t \in [0,T]\ .
\end{align}
We refer to the process $D^{-1}\cS = (1,D^{-1}S)$ as the discounted price process.
We define the process $\Theta = \{(M_t,N_t) :0\leq t \leq T\}$ as the number of shares that the agent owns of each asset at any time $t\in[0,T]$. 
Assuming that, in the market, short-selling and borrowing money from the bank is allowed, there is no restriction on the values assumed by the process $\Theta$. 
To guarantee the existence of the \Ito integral $(\Theta\cdot \cS)_T$, we are going to work with a restricted class of strategies, the ones satisfying the following definition.
\begin{definition}\label{integrability.condition}
Given the filtration $\bT=\{\cT_t, t\geq0\}$, we define an \emph{allowed portfolio} process as a
progressive $\cT_t\text{-measurable process }\Theta^{\bT}(t,\omega):[0,T]\times \Omega \longrightarrow \bR^2$ 
which satisfies
$\int_0^T \norm{\Theta^{\bT}_t \sigma(t,S_t)}^2 dt< +\infty$
almost surely. When $\bT=\bF$, we suppress the superscript notation.
\end{definition}

The total wealth of an \emph{allowed portfolio} process $\Theta$ is then defined as the following \Ito integral
\begin{subequations}\label{def.wealth}
    \begin{align}
        X^{\Theta}_t 
        &= M_t \, D_t + N_t \, S_t \label{def.wealth.sum}\\
        &= X_0 +(\Theta\cdot \cS)_t  \label{def.wealth.int}
        \quad t \in [0,T] \ . 
    \end{align}
\end{subequations}
The second equality indicates that $\Theta$ satisfies the \emph{self-financing condition}.

\subsection{Classical arbitrage conditions}
In this section we are going to define the concepts of arbitrage with the special care of denoting 
the set of strategies that the agent is restricted to use.
We do so by explicitly including in the definition of the arbitrage conditions the set of allowed strategies. 
In the following we use the notation $\cA(\bT)$ to denote a general set of $\bT$-progressive admissible strategies, later we will replace it with more concrete examples. Again when $\bT=\bF$, we omit to indicate the filtration in the notation.

For the following, we need some technical definitions. 
We refer to $L^{\infty}(\Omega,\cF_T,\PP)$ as the set of $\cF_T$-measurable bounded random variables, equipped with the essential supremum  norm $\norm{X}_{\infty}$. 
We abbreviate it to $L^{\infty}$ when the measure $\PP$ is established.
Finally we use $L^{\infty}_+(\Omega,\cF_T,\PP)$ to denote the class of non-negative $\cF_T$-measurable bounded random variables and $L^{0}_+(\Omega,\cF_T,\PP)$ for the class of $\cF_T$-measurable non-negative random variables.
Following the classical definition in \cite[Definition 2.8]{Schachermayer1994a}, we define \emph{the set of contingent claims $\cA$-attainable at price 0},
$$ K(\cA) = \Big\{  (\Theta\cdot \cS)_T  |\ \Theta\in\cA \Big\} = \{    X_T^{\Theta} - X_0 |\ \Theta\in\cA \} \  \ .$$
This set contains all possible random values of terminal wealth that an agent with initial $0$ wealth can reach at time $T$ by only applying $\cA$-admissible strategies.
Then, by allowing the possibility of ``throwing away money'' at time $T$, we define the sets 
\begin{align*}
    C_0(\cA) &= K(\cA)-L_+^0 =  \{ (\Theta\cdot \cS)_T -f\ |\ \Theta \in \cA, \ f\geq 0\ \text{ finite}\},\\ 
    C(\cA) &= C_0(\cA)\cap L^{\infty}.
\end{align*}
together with the set $\overline{C(\cA)}$ as the closure of $C(\cA)$ in $L^{\infty}$.

Given the above, we are ready to state the most important definition in Arbitrage Theory. A complete review of this theory is given in \cite{DelbaenSchachermayer2006}.
\begin{definition}\ \label{def.NA.NFLVR}
Given a semi-martingale $\cS$ and a set of admissible strategies $\cA$ we say that 
\begin{enumerate}[label=(\roman*)]
    \item[ $(NA)_{\cA}$:] the condition of \emph{No Arbitrage} holds on $\cA$ when $K(\cA)\cap L_+^{\infty}=\{0\}$. 
    \item[ $(NFLVR)_{\cA}$:] the condition of \emph{No Free Lunch with Vanishing Risk} holds on $\cA$ when $\overline{C(\cA)}\cap L_+^{\infty}=\{0\}$.
\end{enumerate}
We also define the conditions of \emph{Arbitrage}, $(A)_{\cA}$, and 
\emph{Free Lunch with Vanishing Risk}, $(FLVR)_{\cA}$, as the complements of the conditions above, 
that is $\overline{(NA)}_{\cA}$ and  $\overline{(NFLVR)}_{\cA}$ respectively.
\end{definition}
Since  $K(\cA)\subset \overline{C(\cA)}$, it immediately follows that (NFLVR) implies (NA).

With the following definitions, we describe the main sets of strategies that an agent is allowed to use. 
The reason to consider different sets of strategies is mainly due to the fact that some of them allow the agent to have negative total wealth for some moments in time 
-- condition that we refer to as \emph{temporary-bankruptcy} --,
while others do not allow for this possibility.
\begin{definition}
An allowed portfolio process $\Theta$ is called \emph{a-admissible}, with $a>0$, if 
$(\Theta\cdot \cS)_t\geq -a,\ \forall t\in[0,T]$ almost surely.
If last inequality is strict we call $\Theta$ \emph{super $a$-admissible}, and we simply call it \emph{admissible} if it is \emph{a-admissible} for some $a>0$. 
\end{definition}
\begin{definition}
The set of all admissible (resp. $a$-admissible) strategies is denoted by $\cH$ (resp. $\cH_a$). We write $\cH_a^*$ for the set of super $a$-admissible strategies. 
In particular, we use $\cH^+$ for the set $\cH_{X_0}^*$.
\end{definition}
We will mainly work with the set $\cH$ and $\cH^+$, the latter being the one that forbids the condition of temporary-bankruptcy.

In the following we recall the central result of the theory of pricing and hedging by no-arbitrage, that relates the condition of no arbitrage with the existence of an \textit{Equivalent Local Martingale Measure} (ELMM). A proof of this result, known as the \emph{Fundamental Theorem of Asset Pricing}, is given in quite general settings in \cite{Schachermayer1994a}.

\begin{definition}
    The semi-martingale $S$ satisfies (ELMM) if there is a probability measure $\QQ$ on $\cF$ equivalent to $\PP$ such that the discounted price process is a local martingale with respect to $(\Omega,\cF,\QQ)$.
\end{definition}
    
\begin{proposition}[Fundamental Theorem of Asset Pricing]\label{nflvr.ELMM}
    The semi-martingale $S$ satisfies $(NFLVR)_{\cH}$ if and only if it satisfies (ELMM).
\end{proposition}
The following result on change of probabilities provides a necessary condition for the existence of an ELMM and it will be useful in the context of enlargements of filtration introduced in the next section.
\begin{proposition}[Cameron-Martin-Girsanov]\label{novikov}
 Let $\theta=\prT{\theta}$ be an $\bF$ predictable process, and  such that
 \begin{align}\label{novikov.condition}
      \EE\left[ \exp\left( \frac{1}{2}\int_0^T\theta^2_t \, dt \right) \right] < +\infty\ \ \ \text{(Novikov's Condition)}\ ,
 \end{align}
then there exists a measure $\QQ$ on $(\Omega,\cF_T)$, equivalent to $\PP$, such that
     $$\frac{d\QQ}{d\PP} =  \mathcal{E}\left(\int_0^T \theta_t dB_t\right)$$
    and  $W=(B_t - \int_0^t \theta_u \, du, 0 \leq t \leq T)$ is a $(\bF,\QQ)$ Brownian motion.
    The process $B=(B_t, 0 \leq t \leq T)$ is the $(\bF,\PP)$ Brownian motion appearing in \eqref{asset.risky}.
\end{proposition}
\subsection{Enlargement of filtrations} 
In this section we introduce the concepts required to model the portfolio of an insider agent. 
We assume that an insider has at her disposal more information than the one freely accessible, and we model this
by enlarging the filtration with respect to which she can look for adapted strategies.

To this end, we introduce the filtration $\bG=\{\cG_t, t\geq0\}$ that we assume larger than $\bF$, that is  $\bF\subset\bG$. 
In particular we focus on the case where the additional information is accessible since the initial time, that is $\bG=\bF^{\,\sigma(G)}$ where 
\begin{equation}\label{bG.def}
    \cF^{\,\sigma(G)}_t = \bigcap_{s>t}\left( \cF_s \vee \sigma(G) \right)\ ,
\end{equation}
with $G\in\cF_T$ being a real random variable modeling the privileged information. 

In order to assure that any $\bF$ semi-martingale is also a $\bG$ semi-martingale,
what is known in the literature as \emph{hypothesis (H')}, see \cite{JeulinYor1978},
for the rest of this paper we state the following standing assumption, known as  \emph{condition (A')} in \cite{Jacod1985}, that is stronger as it implies the hypothesis (H'). 

\textbf{Assumption:}
The distribution of $G$ is positive and $\sigma$-finite while the regular $\cF_t$-conditional distributions
almost surely verifies for $t\in[0,T)$ the absolutely continuity condition
$\PP(G \in \cdot|\cF_t) \ll \PP(G \in \cdot)$.

The above assumption assures the existence of a jointly measurable process 
$\eta^g = \prTmenos{\eta^g}$, with $(g,t)\in\bR\times[0,T)$ such that $\PP(A|\cF_t) = \int_A \eta^g_t \PP(G\in dg)$ for any $A\in\sigma(G)$.

The following result allows to compute the $\bG$-semi-martingale decomposition 
of an $\bF$-semi-martingale. Its proof is given in \cite{Jacod1985}.

\begin{proposition}[Jacod, 1985]\label{drift.decomposition}

 There exists a jointly measurable process 
 $\alpha^g = \prTmenos{\alpha^g}$, with $(g,t)\in\bR\times[0,T)$
 such that,
 \begin{itemize}
     \item $\int_0^t (\alpha_u^G)^2 du<+\infty $ almost surely. 
     \item $\langle\eta^G,B\rangle_t =\int_0^t \eta_u^G\alpha_u^G du$.
 \end{itemize}
 and  $W_t = B_t -\int_0^t \alpha_u^G du$ is a $\bG$-Brownian motion.
\end{proposition}

\begin{remark}\label{drift.computation}
Using the second statement of the previous lemma, for $t<T$, we can write
\begin{equation}\label{Doleans-Dade}
\eta_t^G = \mathcal{E} \left(-\int_0^t \alpha_u^G dB_u  \right) \ ,
\end{equation}
where $\mathcal{E}(\cX)$ denotes the \emph{Dol\'eans-Dade exponential} of the semi-martingale $\cX$,
see also \cite{ImkellerFollmer93,baudoin2002}.
\end{remark}

The following proposition gives a first relation between the (FLVR) arbitrage condition and the  $\alpha^G$ process. For a proof see \cite{AnkirchnerImkeller05}.
\begin{proposition}\label{FLVR.and.alpha}
    If $\PP\left( \int_0^T (\alpha^G_t)^2 dt = \infty \right)>0$, then $\cS$ satisfies $(FLVR)_{\cH(\bG)}$.
\end{proposition}

Using results from the previous section combined with the semi-martingale decomposition given in Proposition \ref{drift.decomposition} we can construct a simple test to check the (NFLVR)$_{\cH(\bG)}$ condition in presence of initial enlargement of filtration. This result can be also found in \cite[Theorem 2.3]{IMKELLER2000}.
\begin{corollary}\label{cor:novikov.condition.bG}
Let $G$ as in \eqref{bG.def}, if the process $\alpha^{G}=\prTmenos{\alpha^G}$ satisfies 
\begin{equation}\label{novikov.condition.bG}
    \EE\left[ \exp\left( \frac{1}{2}\int_0^T(\alpha_t^G)^2 dt \right) \right] < +\infty \ ,
\end{equation}
 then $(NFLVR)_{\cH(\bG)}$ holds true.
\end{corollary}
\begin{proof}
By the Fundamental Theorem of Asset Pricing, Proposition \ref{nflvr.ELMM}, the (NFLVR) condition follows by the existence of an ELMM. Combining the Cameron-Martin-Girsanov theorem  (Proposition \ref{novikov}) with the semi-martingale decomposition given in Proposition \ref{drift.decomposition}, the $(NFLVR)_{\cH(\bG)}$ follows by equation  \eqref{novikov.condition.bG} that in this context is equivalent to the Novikov condition \eqref{novikov.condition}. 
\end{proof}

\section{Utility maximization problem}\label{sec:utility.max.prob}
In this section we begin to study the relationships between the set of strategies that the agent is allowed to employ with her maximum expected profit and in general with the conditions of arbitrage.
We start by introducing the general class of utility functions, and then by considering the associated maximal expected utility. 
The utility functions are chosen to satisfy the classical assumptions, that is to be increasing, continuous, differentiable and strictly concave. These assumptions have perfect sense economically. 

\begin{definition}\label{def:utility.functions}
Let $\bT=\{\cT_t, t\geq0\}$ be a given filtration and $U_{\gamma}(x) = (x^{\gamma}-1)/\gamma+\gamma$, with $\gamma\in[0,1]$, an utility function -- for $\gamma=0$ we actually mean the right limit as $\gamma \to 0$, that is $U_0(x)=\ln(x)$ --, we denote by
$$v_\gamma^\bT(\cA) = \sup_{\Theta\in\cA(\bT)}\EE\left[U_\gamma(X_T^{\Theta})\right] 
\quad \gamma\in[0,1] \ ,$$
the corresponding maximal expected utility of the agent constrained to work with the strategies belonging to the set $\cA$. For the extreme cases we use the notations $u^\bT(\cA):=v_0^\bT (\cA)$ and $v^\bT(\cA)= v_1^\bT(\cA)$.
Following the convention adopted above, we omit in the notation to include the filtration $\bT$ whenever the underlying one is $\bF$.
\end{definition}
\begin{remark}
The utility functions as introduced above are a slight modification of the classical Constant Relative Risk Aversion (CRRA) functions, see \cite{merton1969, Bauerle}, by the additional constant term $\gamma$. 
We use this modified form as they preserve the same optimal portfolio while  including the improper linear utility function for $\gamma=1$.
\end{remark}


\begin{proposition}\label{comparison.gain}
With the previous notation, the following inequalities hold,
$$u^\bT(\cH^+) < v_\gamma^\bT(\cH^+) < v^\bT(\cH^+) \leq v^\bT(\cH)$$
\end{proposition}
\begin{proof}
The first two inequalities follow by the fact that $U_\gamma(x)$ is strictly increasing in the parameter $\gamma$ while the last one by the fact that $\cH^+\subset\cH$.
\end{proof}

It can be shown that with a sufficiently large initial capital, arbitrage opportunities do not depend on the bankruptcy, as the following result shows.

 \begin{Lemma}\ \label{lem:arbitrage.comparison}
 Given a filtration $\bT$, the following equivalence hold
 $$
    (NA)_{\cH(\bT)}\Longleftrightarrow \{(NA)_{\cH^+(\bT)},\ \forall X_0\in\bR^+\} \ .$$
\end{Lemma}
\begin{proof}
The necessary condition follows from the inclusion $\cH^+ \subset \cH$. For the sufficient one, we start by assuming that $(NA)_{\cH^+(\bT)}$ holds true for all $X_0\in\bR^+$ and we prove it by contradiction.
If $(A)_{\cH(\bT)}$ holds, then there exists a strategy $\Theta\in\cH(\bT)$ with the property that $(\Theta\cdot \cS)_T \geq 0$ almost surely with $\PP((\Theta\cdot \cS)_T>0)>0$. 
By definition of the set $\cH(\bT)$ we deduce that there exists a constant $a>0$ such that, $\forall\ t\in[0,T]$, $(\Theta\cdot \cS)_t\geq -a+\epsilon$ almost surely, with some $\epsilon > 0$. 
So considering the portfolio
$X_t = a + (\Theta\cdot \cS)_t$, we see that we get $X_T > a$ almost surely 
therefore implying that $(A)_{\cH^+(\bT)}$ holds for $X_0=a$ and $\Theta\in\cH^+_{\bT}$, that is a contradiction. 
\end{proof}

We present a technical result that we will use later. Its proof is given in \cite{AnkirchnerImkeller05}.

\begin{Lemma}\label{lem:admisible.strategy}
Let $\cS = (D,S)$ be the pair of continuous semi-martingale satisfying $(NFLVR)_{\cH(\bT)}$.
If $(\Theta\cdot\cS)_T \geq -X_0$ almost surely with $\Theta$ allowed, then the process $\Theta\in\cH_{X_0}$.
\end{Lemma}
Finally, we get to the result that relates the condition of (NFLVR)$_{\cH(\bT)}$ with the expected terminal utility under $\cH^+(\bT)$ and $\cH(\bT)$. 
\begin{proposition}
Let $U$ be an utility function such $\sup_{x\in\bR}\{U(x)=-\infty\} = 0$, then the following implications hold.
$$  
(NFLVR)_{\cH(\bT)} \Longrightarrow \sup_{\Theta\in\cH(\bT)}\EE\left[U(X_T^{\Theta})\right] = \sup_{\Theta\in\cH^+(\bT)}\EE\left[U(X_T^{\Theta})\right].$$
\end{proposition}
\begin{proof}
    We assume there exists $\Theta\in\cH$ that $\EE[U(X_T^{\Theta})]>\sup_{\Theta\in\cH^+(\bT)}\EE\left[U(X_T^{\Theta})\right]$. 
    Then, $\Theta\not\in\cH^+$ and $X_t^{\Theta}\not\geq 0$ almost surely for some $t\in[0,T]$. 
    Like (NFLVR)$_{\cH(\bT)}$ holds, lemma \ref{lem:admisible.strategy} says that $X_T^{\Theta} < 0$ with positive probability and we conclude that $\EE[U(X_T^{\Theta})]=-\infty$ follows.
\end{proof}

\section{Examples}\label{sec:examples}
In this section, we show several examples to highlight the differences that may or may not exist by playing with one set of strategies or another. 
We focus on the following specific model, deeply studied in the literature, see in particular \cite{merton1969, karatzas1996}, in where the risky asset is given by a Geometric Brownian motion. 
In equation \eqref{asset.risky}, we set the drift and volatility processes to $\mu(t,S_t) = \eta_t S_t$ and $\sigma(t,S_t) = \xi_t S_t$, so that the resulting model is
\begin{subequations}\label{DS.SDE}
\begin{align}
dD_t &= D_t \, r \, dt, 
\label{D.SDE}\\
dS_t &= S_t \, \left(\eta_t dt + \xi_t dB_t\right)   \label{S.SDE}
\end{align}
\end{subequations}
where $r>0$ is the constant interest rate. The processes 
$\eta=\prT{\eta}$ and $\xi=\prT{\xi}$, that we refer to as the proportional drift and volatility, are assumed to be adapted to the natural filtration of the process $B$ with $\eta$, $\xi$ and $1/\xi$ bounded.
Whenever the agent plays with $\Theta\in\cH^+$, the wealth process $X$ is almost surely strictly positive, 
and therefore we can use an alternative form to express the SDE \eqref{def.wealth.int}, that is
\begin{align}\label{def.wealth.pos}
    \frac{dX_t}{X_t} = (1-\pi_t)\frac{dD_t}{D_t} + \pi_t\frac{dS_t}{S_t}.
\end{align}
where $\pi_t = S_t N_t / X_t $ and $N_t$ is defined in \eqref{def.wealth.sum}. 
The following result, proved in \cite{merton1969} and adapted here to our notation, gives the optimal strategies and the corresponding maximal expected utilities for the cases when the agent has no additional information (i.e. under the filtration $\bF$) and works with the strategies that do not allow bankruptcy (i.e. in the set $\cH^+$).
For more details we refer the reader to \cite{karatzas1996}.

\begin{theorem}[Merton, 1969]\label{th:merton1969}
Under the filtration $\bF$ the optimal strategy is
$$
\argsup_{\Theta\in\cH^+}\EE\left[ U_\gamma(X_T^{\Theta}) \right] 
= \frac{1}{1-\gamma}\frac{\eta_t-r}{\xi^2_t} \ , \quad  \gamma\in[0,1)
$$

 and the maximal expected utility is given, with $\gamma\in[0,1)$, by
 \begin{subequations}\label{merton.gain}
 \begin{align}
    u(\cH^+) 
    &= \ln X_0 + rT + \frac{1}{2} \int_0^T\EE\left[ \left(\frac{\eta_t-r}{\xi_t}\right)^2 \right] \, dt 
    & \text{ logarithmic utility}\ , \label{merton.gain.log}\\
    v_\gamma(\cH^+) 
    &= \frac{X_0^{\gamma}}{\gamma}\EE\left[\exp\left(\gamma rT +\frac{1}{2} \frac{\gamma}{1-\gamma} \int_0^T\left( \frac{\eta_t-r}{\xi_t} \right)^2 dt\right)\right] + \frac{\gamma^2-1}{\gamma}
    & \text{ CRRA utility}\ . \label{merton.gain.CRRA}
 \end{align}
\end{subequations}
\end{theorem}

In presence of insider information, as expected, the optimal strategies change since the agent may take advantage of the additional information she has privileged access to.
The following result computes the same quantities as in Theorem \ref{th:merton1969} but under the initial enlarged filtration $\bG$. The process $\alpha^G=\prTmenos{\alpha^G}$ appearing in the statement comes from the semi-martingale decomposition given in Proposition \ref{drift.decomposition}. The result for the logarithmic utility has been proved in \cite{karatzas1996} while the one for general CRRA utilities can be obtained by solving the corresponding HJB equation. Details can be found, for example, in \cite{Bauerle}.
\begin{theorem}
Under the filtration $\bG$ the optimal strategy is
$$
\argsup_{\Theta\in\cH^+(\bG)} \EE\left[ U_\gamma(X_T) \right] 
        = \frac{1}{1-\gamma}\left(\frac{\eta_t-r}{\xi^2_t} + \frac{\alpha^G_t}{\xi_t}\right) \ , \quad 
        \gamma\in[0,1)
$$
 and the maximal expected utility is given, with $\gamma\in[0,1)$, by
 \begin{subequations}\label{insider.gain}
 \begin{align}
    u^\bG(\cH^+) 
    &= 
    \ln X_0 + rT + \frac{1}{2}\int_0^T\EE\left[ \left(\frac{\eta_t-r}{\xi_t}\right)^2 \right] \, dt 
        + \frac{1}{2}\int_0^T\EE\left[ (\alpha^G_t)^2 \right] \, dt 
    & \makebox[0pt][r]{logarithmic utility}\ , \label{insider.gain.log}\\
    v_\gamma^\bG(\cH^+) 
    &=\frac{X_0^{\gamma}}{\gamma}\EE\left[\exp\left(\gamma rT +\frac{1}{2} \frac{\gamma}{1-\gamma} \int_0^T\left( \frac{\eta_t+\alpha^G_t\xi_t-r}{\xi_t} \right)^2 dt\right)\right] + \frac{\gamma^2-1}{\gamma}
    & \text{ CRRA utility}\ . \label{insider.gain.CRRA}
 \end{align}
\end{subequations}
\end{theorem}

From the result above it follows that $u^\bG(\cH^+) <+\infty$ if and only if $\int_0^T\EE\left[ (\alpha_t^G)^2 \right]dt<+\infty$. 

As it is known, for the natural filtration $\bF$, an ELMM can be found and we conclude that NFLVR$_{\cH(\bF)}$ holds true. In the following result, we specify if the maximal expected utility of the agent that plays with the filtration $\bF$ for the different utilities that we have defined is bounded or not.
\begin{proposition}\label{finiteness.under.bF}
Under the modeling assumptions \eqref{DS.SDE}, $v_\gamma^\bF(\cH^+)<+\infty$ for $0\leq \gamma < 1$ and $v^\bF(\cH) = v^\bF(\cH^+) = +\infty$.
\end{proposition}
\begin{proof}
For $0\leq\gamma<1$ the result follows by the boundedness of the processes $\eta$ and $1/\xi$ and the following expression
\begin{equation}
v_\gamma({\cH^+}) = \frac{X_0^{\gamma}}{\gamma}\exp(\gamma r T) \, \EE\left[ \exp\left( \frac{1}{2}\frac{\gamma}{1-\gamma} \int_0^T\left( \frac{\eta_t-r}{\xi_t} \right)^2dt \right)\right]  + \frac{\gamma^2-1}{\gamma} <+\infty\ .
\end{equation}
For $\gamma = 1$, using \emph{Jensen's inequality}, we get
\begin{equation}
    v({\cH^+}) = \lim_{\gamma\to 1}v_\gamma({\cH^+}) \geq \lim_{\gamma\to 1} \frac{X_0^{\gamma}}{\gamma}\exp\left( \gamma rT +\frac{1}{2} \frac{\gamma}{1-\gamma}\EE\left[ \int_0^T\left( \frac{\eta_t-r}{\xi_t} \right)^2dt \right]\right) + \frac{\gamma^2-1}{\gamma} = +\infty\ . 
\end{equation}
\end{proof}
\subsection{Example of (A)$_\bG$ with \texorpdfstring{$v_\gamma^\bG(\cH^+)=\infty$ for $\gamma\in[0,1]$}{v(gamma)(H^+)=infinity for all gammas}}
\label{sec:example.1}

Fixing $G=B_T$, we consider the enlargement of filtration 
$\bG = \bF^{\,\sigma(B_T)}$.
This implies that the insider agent knows since the time $t=0$ the final value $B_T(\omega)$ for any $\omega\in\Omega$. 
The semi-martingale decomposition in the filtration $\bG$ is
\begin{align}\label{semimartingale.decomposition1}
    dB_t =  \alpha_t^{B_T}dt + dW_t \ ,
    \quad \text{ where } \quad \alpha_t^{B_T} = \frac{B_T-B_t}{T-t}\ , \quad 0 \leq t < T
\end{align}
and the process $W=\prT{W}$ is a $\bG$-Brownian motion. 
Applying \Ito calculus, if we use logarithmic utility and the strategy set $\cH^+(\bG)$, equation \eqref{def.wealth.pos} has the following exact solution
$$ \ln X_T = \ln X_0 + \int_0^T \left( r(1-\pi_t) + \eta_t\pi_t +{\pi_t\xi_t}\frac{B_T-B_t}{T-t}-\frac{1}{2}\pi_t^2\xi_t^2 \right) dt + \int_0^T \pi_t\xi_tdW_t \ .$$
By pointwise maximizing the expectation of the equation above, 
we get the optimal strategy
\begin{align}\label{strategy.precise.value}
    \pi_t^*(B_T) = \argsup_{\pi_t \in \cH^+(\bG)}\EE[\ln X_T] = \frac{\eta_t - r}{\xi^2} + \frac{1}{\xi_t}\frac{B_T-B_t}{T-t}
\end{align}
that implies the following pathwise capital gain
\begin{align}\label{opt.wealth.karatzas}
     X_T = X_0\exp\left(rT + \frac{1}{2}\int_0^T(\pi_t^*(B_T) \xi_t)^2dt +\int_0^T \pi_t^*(B_T) \xi_tdW_t  \right)
\end{align}
\begin{proposition}\label{prop.inf.u.case}
Let $G = B_T$, then $v^\bG(\cH) = v_\gamma^\bG(\cH^+) = +\infty$,\ $\forall \gamma\in[0,1]$.
In addition the stronger result, $\lim_{t \to T} \ln(X_t^{\pi^*(B_T)}) = +\infty$, holds almost surely.
\end{proposition}
\begin{proof}
 In \cite{karatzas1996}, it was proved that $\EE\left[\ln X_T^{\pi^*(B_T)}\right]$ is not bounded. 
 Using Proposition \ref{comparison.gain} we get the result for the first statement.

 For the second statement, we denote the function $\alpha_t^{B_T}$ in \eqref{semimartingale.decomposition1} by simply $\alpha_t$ and we use the definition $\beta_t= (\eta_t - r)/(\xi_t)$. 

From \eqref{opt.wealth.karatzas} we have that, for $0<s<t<T$,
\begin{align}
    \log \frac{X_T}{X_0 e^{rT}} 
    &= \frac{1}{2} \int_0^t (\beta_s^2 - \alpha_s^2) \, ds
            + \int_0^t (\beta_s + \alpha_s) \, dB_s \label{int.repr.log.Xt}
\end{align}

Using the semi-martingale decomposition \eqref{opt.wealth.karatzas}, see also \cite{ito1976} and \cite[Chapter VI]{Protter2005}, we have
\begin{equation}
    \int_0^t \alpha_s \, dB_s 
    = \int_0^t \frac{B_T-B_s}{T-s} \, dB_s 
    = B_T \int_0^t \frac{1}{T-s} \, dB_s
   - \int_0^t \frac{B_s}{T-s} \, dB_s \ . \label{int.alpha.dBs}
\end{equation}
By applying the \Ito's lemma to the functions $x/(T-t)$ and $x^2/(T-t)$ respectively, we get
\begin{align*}
    \int_0^t \frac{1}{T-s} \, dB_s 
    =& \frac{B_t}{T-t} - \int_0^t \frac{B_s}{(T-s)^2} \, ds \\
    \int_0^t \frac{B_s}{T-s} \, dB_s 
    =& \frac{1}{2}\frac{B_t^2}{T-t} - \frac{1}{2} \int_0^t \frac{B_s^2}{(T-s)^2} + \frac{1}{T-s} \, ds
\end{align*}
that substituted in \eqref{int.alpha.dBs} give
\begin{align}
    \int_0^t \alpha_s \, dB_s 
    &= \frac{1}{2} \int_0^t \frac{B_s^2 - 2 \, B_T \, B_s}{(T-s)^2} \, ds
    - \frac{1}{2} \frac{B_t^2 - 2 \, B_T \, B_t}{T-t} 
    + \frac{1}{2} \int_0^t \frac{1}{T-s} \, ds \ . \label{int.alpha.dBs.explicit}
\end{align}

\iftrue
In addition
\begin{align*}
    \frac{1}{2} \int_0^t \alpha_s^2 \, ds 
    &= \frac{1}{2} \int_0^t \frac{B_s^2 - 2B_T\, B_s+B_T^2}{(T-s)^2} \, ds 
\end{align*}
\fi

By substituting \eqref{int.alpha.dBs.explicit} in \eqref{int.repr.log.Xt} and by manipulating the expression we finally get
\begin{align*}
    \log \frac{X_T}{X_0 e^{rT}} 
    =& \frac{1}{2} \int_0^t \beta_s^2 \, ds
    + \int_0^t \beta_s \, dB_s + \frac{1}{2}\frac{B_T^2}{T}
    - \frac{1}{2} \frac{(B_T - B_t)^2}{T-t}  
    + \frac{1}{2} \int_0^t \frac{1}{T-s} \, ds \ . \label{int.repr.log.Xt.explicit}
\end{align*}

We have $(B_T - B_t)^2 / (T-t) \sim \chi^2_1$, so all terms but the integral on the right end side have well-defined limits as $t \to T$. 
Since the integral diverges almost surely we get the result.
\end{proof}

The fact that the optimal expected logarithmic utility is almost surely unbounded implies that the filtration $\bG$ admits arbitrage opportunities.
However the strategy $\pi^*(B_T)$ is quite complicated to implement as, according to \eqref{strategy.precise.value}, it is of finite quadratic variation,
while for this example, it is easy to see that the arbitrage can be achieved more easily by employing strategies with a more clear financial meaning, such as buy, hold and sell. 
We give in the following a precise definition of this kind of strategies as given in \cite[Definition 7.1]{Schachermayer1994a}.
\begin{definition}
A simple predictable strategy is a linear combination of processes of the form \hbox{$\Theta_t=M\bOne_{(T_1,T_2]}$}, where $M$ is $\cF_{T_1}$ measurable and $T_1$ and $T_2$ are finite stopping times with respect to the filtration $\cF_{T}$.
\end{definition}
When stopping time $T_1$ happens, the agent \emph{buys} a quantity $M\in\cF_{T_1}$ of shares of the risky asset at price $S_{T_1}$, borrowing the corresponding money from the riskless asset.
Until $T_2$, she \emph{holds} her position and then she \emph{sells} her shares at price $S_{T_2}$.

For the following proposition, we assume that the proportional volatility process is constant and the proportional drift process is deterministic.
Under this simplifying assumption, the insider information on $B_T$ is equivalent to the one on $S_T$ according to the following equation
$$ S_T = S_0\exp\left(\int_0^T\left(\eta_t-\frac{\xi^2}{2}\right)dt + \xi B_T\right) \ .$$
\begin{proposition}\label{prop:strategy.semi-infinite.inteval}
    Let $G = B_T$ and $X_0>0$, then $(A)_{\cH^+(\bG)}$ holds.
\end{proposition}
\begin{proof}
Let $0<\epsilon < S_T$ and we define the following stopping time.
\begin{equation}\label{stopping.time}
    \tau_\epsilon = \inf\{ S_t\ < e^{-r(T-t)}\left(S_T - \epsilon\right) \} \ .
\end{equation}
On the set $\{\tau_\epsilon<T\}$, the agent invests her money in the risky asset at time $\tau_\epsilon$. 
This strategy is modeled by,
$$ N_t = X_0 \frac{e^{r \, \tau_\epsilon}}{S_{\tau_\epsilon}}\bOne\{t\geq \tau_\epsilon\}. $$
On $\{\tau_\epsilon < T\}$, $X_T = e^{rT} X_0 \frac{S_T}{S_T-\epsilon} $ and on its complement $X_T = e^{rT}X_0$. 
As $\PP(\tau_\epsilon<T)>0$, we conclude that $X_T\geq e^{rT}X_0$ almost surely and $\PP(X_T > e^{rT}X_0) >0.$
\end{proof}
The Figure \ref{fig:tau} shows an example of the situation that is described in the proof. 
When the stopping time happens, the insider trader invests in the stock 
as she knows almost surely that she will realize a positive profit.
\begin{figure}[H]\centering
    \includegraphics{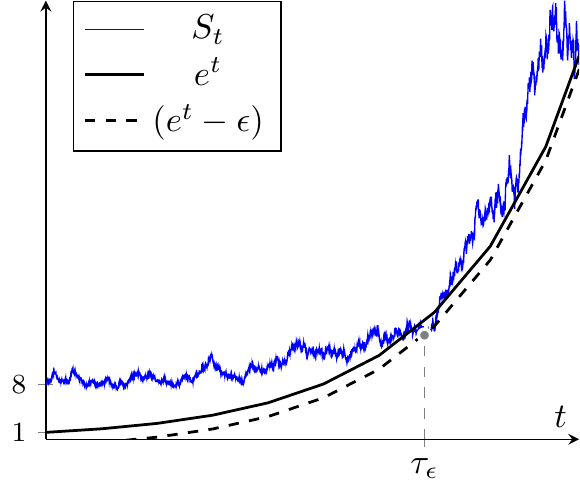}
    \caption{\label{fig:tau}
    An example of realization of the stopping time $\tau_\epsilon$ defined in \eqref{stopping.time}.}
\end{figure}

\begin{corollary}
    Let $G = B_T$, then we have $(A)_{\cH(\bG)},\ (FLVR)_{\cH^+(\bG)}$ and $(FLVR)_{\cH(\bG)}$.
\end{corollary}

\subsection{Example of (A)$_\bG$ with \texorpdfstring{$v_0(\cH^+)<\infty$ and $v_\gamma^\bG(\cH^+)=\infty$ for $\gamma\in(0,1]$}{v(gamma)(H^+)=infinity for all gammas but 0}}
\label{sec:example.2}
Here we analyze another example of enlargement of filtrations. It was introduced by \cite{karatzas1996} to study the case when the insider trader knows a lower or an upper bound of the stock price in a certain future horizon time.
We use the initial enlargement $\bG = \bF^{\,\sigma(G)}$ with
\begin{equation}
\label{G.example.2}
    G = \bOne \{ c_1 \leq B_T\leq c_2 \}\ .
\end{equation}

In the following, we show that the additional information carried by the filtration $\bG$
implies a finite terminal logarithmic utility, and therefore different for the case analyzed in Section \ref{sec:example.1}.

We start by computing the explicit expression of the drift $\alpha^G = \prTmenos{\alpha^G}$ appearing in the $\bG$-semi-martingale decomposition given in Proposition \ref{drift.decomposition}. The proof is deferred to the appendix.

\begin{proposition}\label{alpha.int.dec}
    Let $G$ as in \eqref{G.example.2}, then
    $$\alpha^g_{t} 
    = \frac{(-1)^g}{\sqrt{T-t}}\ \frac{\Phi'\left(\frac{c_2-B_t}{\sqrt{T-t}}\right) - \Phi'\left(\frac{c_1-B_t}{\sqrt{T-t}}\right) }{\PP( G=g \, |B_t)}
    \ , \quad 0 \leq t < T$$
and more explicitly
\begin{equation} 
\alpha^g_{t} =
\left\{\begin{array}{ll}
    \displaystyle
    \frac{1}{\sqrt{T-t}}\ \frac{\Phi'\left(\frac{c_2-B_t}{\sqrt{T-t}}\right) - \Phi'\left(\frac{c_1-B_t}{\sqrt{T-t}}\right) }{\Phi\left(-\frac{c_2-B_t}{\sqrt{T-t}}\right) +\Phi\left(\frac{c_1-B_t}{\sqrt{T-t}}\right)}
    & \text{ when } g = 0 \ ,\\
    \displaystyle \frac{1}{\sqrt{T-t}}\ \frac{\Phi'\left(\frac{c_1-B_t}{\sqrt{T-t}}\right) - \Phi'\left(\frac{c_2-B_t}{\sqrt{T-t}}\right) }{\Phi\left(\frac{c_2-B_t}{\sqrt{T-t}}\right) - \Phi\left(\frac{c_1-B_t}{\sqrt{T-t}}\right)} 
    & \text{ when } g = 1 \ .
\end{array}\right.
\end{equation} 
where $\Phi$ denotes the distribution of a standard Gaussian random variable.
\end{proposition}

Before the main result of this section we need to introduce the following technical lemma whose proof is deferred to the appendix.
\begin{Lemma}\label{lem:bound.I.finite}
The integral of the function $I(x,t)$ defined as 
\begin{equation}
    I(x,t) =\frac{1}{\sqrt{T-t}} \frac{ \left[ \Phi'(z_1)-\Phi'(z_2)\right]^2 }{\left[\Phi\left( z_2  \right)  - \Phi\left( z_1  \right)\right] \left[\Phi\left( -z_2  \right)  + \Phi\left( z_1  \right)\right]}, \label{def.I}
\end{equation}
in the variable $x\in\bR$ is uniformly bounded for $t\in[0,T]$.
In \eqref{def.I}, we have used the following definitions
$z_2 = {(c_2-x)}/{\sqrt{T-t}}$ and 
$z_1 = {(c_1-x)}/{\sqrt{T-t}}$. 
\end{Lemma}

The following result shows that the logarithmic utility optimization allows for a finite optimum. This result was first conjectured in \cite{karatzas1996}, 
where they supported the conjecture via numerical results.
Then the conjecture was solved in the general entropy setting by \cite{AmendingerImkellerSchweizer1998}. Here we give a more direct proof on the line of the arguments given in \cite{karatzas1996}.

\begin{theorem}\label{thm:karatzas1996.conjecture}
    Let $G$ as in \eqref{G.example.2}, then
    $u^{\bG}(\cH^+)<\infty$.
\end{theorem}
\begin{proof}
By using the expression of $\alpha^G$, given in \cite[Equation (4.25)]{karatzas1996}, we have
$$ \EE[(\alpha_t^G)^2] =  \frac{1}{2\pi\sqrt{T-t}\sqrt{2\pi t}} \int_{\mathbb{R}}  I(x,t) \, e^{-x^2/2} \, dx $$
where $I(x,t)$ is defined in \eqref{def.I} and the volatility process is bounded.
By \eqref{insider.gain.log}, it is enough to prove that, for some constant $K>0$,
$$ \EE[(\alpha_t^G)^2] \leq \frac{K}{\sqrt{t(T-t)}} \ . $$
This follows by Lemma \ref{lem:bound.I.finite}.
\end{proof}
\begin{proposition}
    Let $G$ as in \eqref{G.example.2}, then $v_\gamma^{\bG}(\cH^+)=+\infty$ for $\gamma \in [1/2,1]$.
\end{proposition}
\begin{proof}
By Corollary \ref{arbitrage.interval} the $(FLVR)_{\cH(\bG)}$ condition holds and therefore an ELMM can not exist. 
This implies that the Novikov condition, given in \eqref{novikov.condition.bG}, is not satisfied. We conclude that
      $\EE\left[ \exp\left( \frac{1}{2}\int_0^T(\alpha^G)^2 dt \right) \right] = +\infty.$
Moreover, if $\gamma\geq 1/2$, it follows that $\gamma/(1-\gamma)\geq 1$ and we can conclude that
\begin{align*}
  v_\gamma({\cH^+}) &= \frac{X_0^{\gamma}}{\gamma}\exp(\gamma r T)\EE\left[ \exp\left( \frac{1}{2}\frac{\gamma}{1-\gamma} \int_0^T\left( \frac{\eta_t-r}{\xi_t} + \alpha_t^G \right)^2dt \right)\right] + \frac{\gamma^2-1}{\gamma}\\
  &\geq \frac{X_0^{\gamma}}{\gamma}\exp(\gamma r T)\EE\left[ \exp\left( \frac{1}{2} \int_0^T\left( \frac{\eta_t-r}{\xi_t} + \alpha_t^G \right)^2dt \right)\right] + \frac{\gamma^2-1}{\gamma}=+\infty\ . 
\end{align*}
\end{proof}
Now, we are going to enunciate the results on arbitrage for the strategy sets $\cH(\bG)$ and $\cH^+(\bG)$.
Using the solution of the process $S$ at time $T$, under constant proportional volatility and proportional deterministic drift processes, this problem is equivalent to,
\begin{equation}\label{G.example.2.S}
    G = \bOne \Big\{b_1\leq  S_T \leq b_2 \Big\}\ ,\quad b_i = S_0\exp\left(\int_0^T\left(\eta_t - \frac{\xi^2}{2}\right) dt + \xi\ c_i \right) \ . 
\end{equation}

So, under these hypothesis of the coefficient processes, we can link the privileged information between the process $B=\prT{B}$ and $S=\prT{S}$. 
In the following statements, we use the assumption and notation given in \eqref{G.example.2.S}. 

\begin{proposition}\label{prop:strategy.finite.inteval}
    Let $G$ as in \eqref{G.example.2}, then  $(A)_{\cH^+(\bG)}$ holds.
\end{proposition}
\begin{proof}
The result follows by a similar argument of Proposition \ref{prop:strategy.semi-infinite.inteval}.
\end{proof}
\begin{corollary}\label{arbitrage.interval}
Let $G$ as in \eqref{G.example.2},
then we have $(A)_{\cH(\bG)},\ (FLVR)_{\cH^+(\bG)}$ and $(FLVR)_{\cH(\bG)}$.
\end{corollary}
\begin{proposition}
    Let $G$ as in \eqref{G.example.2}, 
    then $v^{\bG}(\cH)=+\infty$.
\end{proposition}
\begin{proof}
Let $0< \epsilon< b_1/2$ and consider the following stopping time $\tau_\epsilon$,
\begin{equation}\label{stopping.time.interval}
    \tau_\epsilon = \inf\{ S_t\ < e^{-r(T-t)}(b_1 - \epsilon) \} \ .
\end{equation}
and the strategy 
$\Theta = \{(M_t,N_t) :0\leq t \leq T\}$ with 
$ N_t = M \, X_0 \frac{e^{r\, \tau_\epsilon}}{S_{\tau_\epsilon}}\bOne\{t \geq \tau_\epsilon\},$
for some constant $M>0$.
On $\{\tau_\epsilon < T\}$, $X_T = e^{rT} X_0\  M\frac{S_T}{b_1- \epsilon}+e^{rT}X_0\ (1-M)$ and on its complement $X_T = e^{rT}X_0$,
hence,
$$ \EE[X_T^\Theta] = e^{rT}X_0  + e^{rT}X_0\  M\left[\frac{S_T}{b_1 -\epsilon}-1\right]\PP(\tau_\epsilon < T), $$
which is not bounded as $M\to\infty$ for a fixed $\epsilon >0$.
\end{proof}

\subsection{Example of (NFLVR)$_\bG$ with \texorpdfstring{$v_\gamma^\bG(\cH^+)<\infty$ for $\gamma\in[0,1)$ and $v_1^\bG(\cH^+) = +\infty$}{v(gamma)(H^+)=infinity only for gamma=1}}
\label{sec:example.3}

In this section we show that the acquisition of additional information by an insider agent does not directly implies that she can take advantage of an arbitrage. Indeed we show that even knowing information about the terminal price of the risky asset (we model this by a function of $B_T$) may not lead to an arbitrage condition.

We use the initial enlargement $\bG = \bF^{\,\sigma(G)}$ with
\begin{equation}
\label{G.example.3}
    G = \bOne \Big\{ B_T \in \cup_{k=-\infty}^{+\infty} [2k-1,2k] \Big\}
\end{equation}
implying that the insider trader only knows if the Brownian motion will end up in a particular infinite union of intervals of size one.
Following the same arguments of Proposition \ref{alpha.int.dec}, next result gives a closed form expression of the new drift of the $\bG$-semi-martingale decomposition. It's proof is deferred to the appendix.
\begin{proposition}\label{alpha.int.unbounded}
    Let $G$ be as in \eqref{G.example.3}, then
    $$\alpha^g_{t} 
    =\frac{(-1)^g}{\sqrt{T-t}\ \PP(G=g \,|B_t)}
    \left(\sum_{k=-\infty}^{+\infty} \Phi'\left(\frac{2k-B_t}{\sqrt{T-t}}\right) - \Phi'\left(\frac{2k-1-B_t}{\sqrt{T-t}}\right) \right)
    \ , \quad 0 \leq t < T$$
and more explicitly
\begin{equation} 
\alpha^g_{t} =
\left\{\begin{array}{ll}
    \displaystyle
 \frac{1}{\sqrt{T-t}}\ \frac{\sum_{k=-\infty}^{+\infty} \Phi'\left(\frac{2k-B_t}{\sqrt{T-t}}\right) - \Phi'\left(\frac{2k-1-B_t}{\sqrt{T-t}}\right) }{\sum_{k=-\infty}^{+\infty}\Phi\left(-\frac{2k-B_t}{\sqrt{T-t}}\right) + \Phi\left(\frac{2k-1-B_t}{\sqrt{T-t}}\right)}
    & \text{ when } g = 0 \ ,\\
    \displaystyle
 \frac{1}{\sqrt{T-t}}\ \frac{\sum_{k=-\infty}^{+\infty} \Phi'\left(\frac{2k-1-B_t}{\sqrt{T-t}}\right) - \Phi'\left(\frac{2k-B_t}{\sqrt{T-t}}\right) }{\sum_{k=-\infty}^{+\infty}\Phi\left(\frac{2k-B_t}{\sqrt{T-t}}\right) - \Phi\left(\frac{2k-1-B_t}{\sqrt{T-t}}\right)}
    & \text{ when } g = 1 \ .
\end{array}\right.
\end{equation}
where $\Phi$ denotes the distribution of a standard Gaussian random variable.
\end{proposition}

Before stating the main result of this section we state the following version of the Mean Value Theorem for definite integrals, whose proof may be found in \cite[Lemma 2.1]{inequality}.
\begin{Lemma}\label{lem:mean.value.thm}
Let $c_1<c_2$, and $f(\cdot)$ a differentiable function. There exists $\xi \in(c_1,c_2)$ such that
\begin{align}
    \frac{1}{c_2-c_1}\int_{c_1}^{c_2} f(y) \, dy &= \notag \frac{f(c_1)+f(c_2)}{2} -\frac{1}{c_2-c_1} \int_{c_1}^{c_2} \left(y- \frac{c_1+c_2}{2}\right) \, f'(y) \, dy \\
    &=\frac{f(c_1)+f(c_2)}{2} - \left(\xi - \frac{c_1+c_2}{2}\right) \, f'(\xi) \ . \label{mean.value.thm}
\end{align}
\end{Lemma}

The following result shows that under the enlargement by the random variable in \eqref{G.example.3}
the insider does not get any possibility of arbitrage. 

\begin{theorem}
     Let $G$ be as in \eqref{G.example.3}, then (NFLVR)$_\bG$ holds.
\end{theorem}
\begin{proof}
By Corollary \ref{cor:novikov.condition.bG}, it is enough to prove that the process $\alpha^G$ satisfies the Novikov condition \eqref{novikov.condition.bG}. 
Using the Jensen inequality we get 
$$  \EE\left[ \exp\left( \frac{1}{2}\int_0^T(\alpha_t^G)^2 dt \right) \right] \leq \frac{1}{T}\int_0^T \EE\left[\exp \left( \frac{T}{2}(\alpha_t^G)^2 \right) \right] \, dt \ ,$$
so we are left to find a bound for the right hand side of the expression above. 
The expectation can be computed as
\begin{equation}
    \EE\left[\exp \left(\frac{T}{2}(\alpha_t^G)^2 \right) \right]
    =\sum_{g\in\{0,1\}}
    \int_{\bR}
    \exp\left( \frac{T}{2}(\alpha_t^g)^2\right)
    \PP(G=g\,|B_t = x) \, 
    d\PP(B_t \leq x) \ .
\end{equation}
Using the similarity of the integrals for $g\in\{0,1\}$, 
we focus on getting a uniform bound, integrable with respect to $t$ for the case $g=1$.

Let $A=\cup_{k=-\infty}^{+\infty} [2k-1,2k]$ be the set appearing  in \eqref{G.example.3}. If $x\in A$, then $\PP(G=1\,|B_t = x)$ is far from zero and the numerator of $\alpha^1_t$ is uniformly convergent as $t \to T$ and bounded by $2$ in the interval $[0,T]$. Bounding the term $\PP(G=1\,|B_t = x)$ by $1$ we get the following inequality 
$$\EE\left[\exp \left(\frac{T}{2}(\alpha_t^1)^2 \right)\right] 
\leq   
K + \int_{A^c}
    \exp\left( \frac{T}{2}(\alpha_t^1)^2\right) \,  
    d\PP(B_t \leq x) \ ,$$
where $K>0$ is some constant.
    
For the moment we focus on the value of the integral in a single interval, say $[c_1, c_2]$, to later generalize the argument to the whole set $A^c$.

In the following we assume that $x \not\in [c_1,c_2]$ and $z_i = (c_i-x)/\sqrt{T-t}$ with $i\in\{1,2\}$. By Lemma \ref{lem:mean.value.thm}, there exists $\zeta\in[c_1,c_2]$ such that
\begin{align}
     \PP\left( B_T\in [c_1,c_2] \,|B_t = x\right)
     =& \frac{c_2-c_1}{\sqrt{T-t}}\frac{\Phi'(z_1)+\Phi'(z_2)}{2} \notag  \\
     &+ \left(\frac{\zeta-x}{\sqrt{T-t}}- \frac{c_1+c_2-2x}{2\sqrt{T-t}}\right) \frac{c_2-c_1}{\sqrt{T-t}}\exp\left(-\frac{(\zeta-x)^2}{2(T-t)}\right) \ .
\end{align}
Using the expression above to find an alternative form of the following term,
\begin{align}
    \frac{1}{\sqrt{T-t}} \frac{\Phi'(z_1)-\Phi'(z_2)}{ \PP\left( B_T\in [c_1,c_2] \,|B_t = x\right)} 
    &= \frac{\left(\Phi'(z_1)-\Phi'(z_2)\right)/\left(\Phi'(z_1)+\Phi'(z_2)\right)}{ {(c_2-c_1)}/2  + \left(\zeta- \frac{c_1+c_2}{2}\right) \frac{c_2-c_1}{\sqrt{T-t}}\frac{\exp\left(-\frac{(\zeta-x)^2}{2(T-t)}\right)}{\Phi'(z_1)+\Phi'(z_2)}}
\end{align}
that allows to show that it is bounded for all $t\in[0,T]$. 
It is enough to look at the limit as $t\to T$. 
The numerator is easily shown to be bounded as
$$ 0\leq \lim_{t\to T} \frac{\Phi'(z_1)-\Phi'(z_2)}{\Phi'(z_1)+\Phi'(z_2)}\leq  \lim_{t\to T} \frac{\Phi'(z_1)}{\Phi'(z_1)+\Phi'(z_2)} \leq 1 \ .$$
Therefore the boundedness follows from the fact that the denominator is sum of a constant term and another one that goes to zero. 
Indeed, assuming w.l.o.g. that $x < c_1 < c_2$ -- the other case being equivalent by the symmetry of the function $\Phi'$ --, we have
\begin{align*}
     \lim_{t\to T} \frac{\frac{c_2-c_1}{\sqrt{T-t}} \exp\left(-\frac{(\xi-x)^2}{2(T-t)}\right)}{\Phi'(z_1)+\Phi'(z_2)} &= \lim_{t\to T}\sqrt{2\pi} \frac{\frac{c_2-c_1}{\sqrt{T-t}} \exp\left(-\frac{(\xi-x)^2}{2(T-t)}\right)}{\exp\left(-\frac{(c_1-x)^2}{2(T-t)}\right)+\exp\left(-\frac{(c_2-x)^2}{2(T-t)}\right)}\\ &=\lim_{t\to T} \sqrt{2\pi} \frac{\frac{c_2-c_1}{\sqrt{T-t}}\exp\left( -\frac{(\xi-x)^2-(c_1-x)^2}{2(T-t)}\right)}{1+\exp\left(-\frac{(c_2-x)^2-(c_1-x)^2}{2(T-t)}\right)} = 0 \ .
\end{align*}

Fixing $t$ sufficiently near to $T$, we have shown that  
$$\frac{1}{\sqrt{T-t}}\frac{\Phi'(z_1)-\Phi'(z_2)}{ \PP\left( B_T\in [c_1,c_2] \,|B_t = x\right)} \leq 1 + \frac{{2}}{c_2-c_1} \ , \quad x\not\in[c_1,c_2] \ .$$

Replacing the interval $[c_1,c_2]$ by the set $A = \cup_{k=-\infty}^{+\infty} [2k-1,2k]$, and by repeating the same argument above, we get that
$$\alpha_t^1 = \frac{\sum_{k=-\infty}^{+\infty} \Phi'\left(\frac{2k-1-x}{\sqrt{T-t}}\right) - \Phi'\left(\frac{2k-x}{\sqrt{T-t}}\right) } {\PP( B_T\in A \,|B_t = x)} \leq 3 
\ , \quad x \in A^c \ .$$
Finally,
$$\EE\left[\exp \left(\frac{T}{2}(\alpha_t^1)^2 \right)\right] 
\leq   
K + \int_{A^c}
    \exp\left( \frac{9}{2}T\right)d\PP(B_t \leq x) < +\infty$$
and the result follows. 
\end{proof}

\begin{proposition}\label{finiteness.under.bG}
Let $G$ be as in \eqref{G.example.3}, then $v_\gamma^\bG(\cH^+)<+\infty$ for $0\leq \gamma < 1$ and $v^\bG(\cH) = v^\bG(\cH^+) = +\infty$.
\end{proposition}
\begin{proof}
Similarly to the arguments used in Proposition \ref{finiteness.under.bF}, 
we can show the result starting from equations \eqref{insider.gain}.
\end{proof}

\section{Conclusions}\label{sec:conclusions}

\begin{table}[h]
    \caption{\label{tab:conc.rmks}Summary of the utility maximization problems analyzed in Section \ref{sec:examples}.}
\begin{center}
    \def\rsp{\hspace*{1cm}}
    \def\lsp{\hspace*{0.15cm}}
\begin{tabular}{lccccccc}
\hline
Sect. & $\bT$ & $u^\bT(\cH^+)$ & $u^\bT(\cH)$ & $v_\gamma^\bT(\cH^+)$  & $v_\gamma^\bT(\cH)$ & $v^\bT(\cH^+) = v^\bT(\cH)$ & Arbitrage cond\\\hline

\ref{sec:examples}  & $\bF$ & \multicolumn{4}{c}{$<+\infty$} & \multicolumn{1}{c}{$=+\infty$} & (NA)$_\bF$ \\ 
\ref{sec:example.3} & $\bG$ & \multicolumn{4}{c}{$<+\infty$} & \multicolumn{1}{c}{$=+\infty$} & (NA)$_\bG$ \\
\ref{sec:example.2} & $\bG$ & \multicolumn{1}{c}{$<+\infty$}  &
\multicolumn{3}{c}{\small (in $\cH^{+}$ for $\gamma\geq 1/2$)} &\multicolumn{1}{c}{$=+\infty$} &
  (A)$_\bG$ \\ 
\ref{sec:example.1} & $\bG$ & \multicolumn{4}{c}{}  &
 \multicolumn{1}{c}{$=+\infty$}  & (A)$_\bG$ \\\hline
\end{tabular}
\end{center}
\end{table}

In this paper we analyzed the relations between the arbitrage conditions, the utility maximization problems and the enlargements of filtration. In particular we considered all these concepts with the respect to the class of strategies the agent may employ in maintaining her portfolio, by focusing on the general admissible class, $\cH$, and the one that does not allow for temporary-bankruptcy, $\cH^+$.

In terms of arbitrage, by Lemma \ref{lem:arbitrage.comparison}, there is practically no difference between working with the class $\cH$ and the class $\cH^+$.
However in terms of utility maximization, the difference becomes clear as we found cases in which the (logarithmic) utility maximization is finite even in presence of arbitrage. 
We have included examples of this type and Table \ref{tab:conc.rmks} shows a brief summary of the results.

In particular it shows with the example of Section 4.3 that an enlargement of initial type does not always imply arbitrage for the insider.

\section*{Acknowledgments}
This research was partially supported by the Spanish Ministry of Economy and Competitiveness grants MTM2017-85618-P (via FEDER funds),  MTM2015-72907-EXP and FPU18/01101. We thank the referees for the insightful comments that helped to improve the presentation of this work.

\section*{Conflict of interest}
The authors declare no conflicts of interest in this paper.

\bibliographystyle{AIMSunsort}
\bibliography{main.bib}

\section*{Appendix}
\begin{proof}[Proof of Proposition \ref{alpha.int.dec}]
From \eqref{Doleans-Dade} it follows that the process $\alpha^G=\prTmenos{\alpha^G}$  satisfies the following relation
\begin{equation}\label{eta.derivative}
    d\eta_t = \eta_t \alpha^G_t dB_t\ ,
\end{equation}
where $\eta$ is the Radon-Nikodym derivative between the law of $G = \bOne \{ c_1 \leq B_T\leq c_2 \}$ conditioned to the the $\sigma$-algebra $\cF_t$ and its unconditional law, i.e,
\begin{equation}\label{eq.eta}
    \eta_t(g) = \frac{\PP(G = g |\cF_t)}{\PP(G = g)}\ ,\quad g\in\{0,1\}\ .
\end{equation} 
Substituting \eqref{eq.eta} in \eqref{eta.derivative} we try to get the an expression for the process $\alpha^G$. 

The conditional probability mass function of $G$ computed for $\{G=1\}$ may be expressed in terms of Gaussian distributions as
\begin{equation}\label{prob.G.by.Gaussian}
\PP(G = 1 |\cF_t) = \Phi\left(\frac{c_2-B_t}{\sqrt{T-t}}\right) - \Phi\left(\frac{c_1-B_t}{\sqrt{T-t}}\right) \ .
\end{equation}

By \Ito's lemma applied to the function $\Phi\left((a-x)/(\sqrt{T-t})\right)$, with $a\in\bR$,  we have that  
\begin{equation}\label{ito.Gauss}
   d\Phi\left(\frac{a-B_t}{\sqrt{T-t}}\right) = -\frac{1}{\sqrt{T-t}}\Phi'\left(\frac{a-B_t}{\sqrt{T-t}}\right) dB_t \ .  
\end{equation}

Substituting \eqref{ito.Gauss} in \eqref{prob.G.by.Gaussian} and applying linearity it follows that 
$$ d\PP(G = 1 |\cF_t) =  \frac{1}{\sqrt{T-t}} \left(\Phi'\left(\frac{c_1-B_t}{\sqrt{T-t}}\right) - \Phi'\left(\frac{c_2-B_t}{\sqrt{T-t}}\right) \right) dB_t $$
and finally,
\begin{equation}
    \frac{d\eta_t}{\eta_t} =  \frac{1}{\sqrt{T-t}} \frac{\Phi'\left( \frac{c_1-B_t}{\sqrt{T-t}}\right) - \Phi'\left(\frac{c_2-B_t}{\sqrt{T-t}}\right)}{\PP(G = 1 |\cF_t)} dB_t\ ,
\end{equation}
and we get the result. The case $\{G=0\}$ follows on the same lines.

\end{proof}

\begin{proof}[Proof of Lemma \ref{lem:bound.I.finite}]
We start by splitting $\bR$ in three intervals  $(-\infty,c_1]$, $(c_1,c_2)$ and $[c_2,\infty)$, then we prove that on each interval the integral is finite.

\textbf{Interval $(-\infty,c_1]$:}
We apply a change of variable in $z_1$ and express $z_2 = z_1 + (c_2-c_1)/\sqrt{T-t}$. 
We let $s_t:=(c_2-c_1)/\sqrt{T-t}$ and call its minimum in $t$ as $s_0>0$. 
We get
\begin{align}
\int_{-\infty}^{c_1}I(x,t)dx  &=  \int_0^{+\infty} \frac{\left[\Phi'(z_1)-\Phi'(z_1 + s_t)\right]^2  }{\left[\Phi\left( z_1 +s_t  \right)  - \Phi\left( z_1  \right)\right] \left[\Phi\left( -z_1 - s_t  \right)  + \Phi\left( z_1  \right)\right]} dz_1 \nonumber \\
&= \int_0^{+\infty} \left( \frac{ \left[\Phi'(z_1)-\Phi'(z_1 + s_t)\right]^2}{\Phi\left( z_1 +s_t  \right)  - \Phi\left( z_1  \right)} + \frac{\left[\Phi'(z_1)-\Phi'(z_1 + s_t)\right]^2}{\Phi\left( -z_1 - s_t  \right)  + \Phi\left( z_1  \right)}  \right) dz_1 \ .
\label{lower.semi.interval}\end{align}
We continue by showing that both terms are finite.
We first consider the first term.
\begin{align*}
\int_0^{+\infty}\frac{ \left[\Phi'(z_1)-\Phi'(z_1 + s_t)\right]^2}{\Phi\left( z_1 +s_t  \right)  - \Phi(z_1)} dz_1 &\leq \int_0^{+\infty}\frac{ \left[\Phi'(z_1)-\Phi'(z_1 + s_t)\right]^2}{\Phi\left( z_1 +s_0  \right)  - \Phi(z_1)} dz_1 \leq \int_0^{+\infty}\frac{ \left[\Phi'(z_1)\right]^2}{\Phi\left( z_1 +s_0  \right)  - \Phi(z_1)} dz_1\\
&= \int_0^{1}\frac{ \left[\Phi'(z_1)\right]^2}{\Phi\left( z_1 +s_0  \right)  - \Phi(z_1)} dz_1 + \int_1^{+\infty}\frac{ \left[\Phi'(z_1)\right]^2}{\Phi\left( z_1 +s_0  \right)  - \Phi(z_1)} dz_1\ .
\end{align*}
The integral in $[0,1]$ is clearly bounded. While for the interval $[1,+\infty]$, we apply a comparison criteria with the function $f(z) = 1/z^2$, as follows,
\begin{align*}
\lim_{z_1\to\infty} z_1^2 \frac{  \left[\Phi'(z_1)\right]^2}{\Phi\left( z_1 +s_0  \right)  - \Phi(z_1)} &= \lim_{z_1\to\infty} \frac{1}{\sqrt[]{2\pi}} \frac{z_1^2\left[\exp\left( -{z_1^2}/{2} \right) \right]^2}{\int_{z_1}^{z_1+s_0}\exp(-u^2/2) \, du}\\
&=\lim_{z_1\to\infty} \frac{1}{\sqrt[]{2\pi}} \frac{2z_1\left[\exp\left( -{z_1^2}/{2} \right) \right]^2 - 2z_1^3 \left[\exp\left( -{z_1^2}/{2} \right) \right]^2}{\exp(-(z_1+s_0)^2/2) - \exp(-z_1^2/2)}\\
&= \lim_{z_1\to\infty} \frac{1}{\sqrt[]{2\pi}} \frac{2z_1\exp\left( -{z_1^2}/{2} \right) - 2z_1^3 \exp\left( -{z_1^2}/{2} \right) }{\exp(-z_1s_0)\exp(-s_0^2/2) - 1}\to 0 \ .
\end{align*}
In the second equality above, we used L'Hopital Rule and we conclude that the integral is finite on $(1,+\infty)$. 

As for the second term in \eqref{lower.semi.interval}, we have the following bound,
\begin{align*}
\int_0^{+\infty}\frac{\left[\Phi'(z_1)-\Phi'(z_1 + s_t)\right]^2}{\Phi(-z_1-s_t)  + \Phi(z_1)} dz_1 &\leq \int_0^{+\infty} \frac{\left[\Phi'(z_1)-\Phi'(z_1 + s_t)\right]^2}{\Phi(z_1)} dz_1 \\
&\leq 2 \int_0^{+\infty}\left[\Phi'(z_1)-\Phi'(z_1 + s_t)\right]^2dz_1\\
&\leq 2 \int_0^{+\infty}\left[\Phi'(z_1)\right]^2dz_1 = \frac{1}{\sqrt[]{2}} \ .
\end{align*}

\textbf{Interval $[c_2,+\infty)$:}
We proceed in the same way as above, now applying a change of variable in $z_2$.
\begin{align}
\int^{+\infty}_{c_2}I(x,t)dx &= \int_{-\infty}^{0}\frac{ \left[ \Phi'(z_2-s_t)-\Phi'(z_2)\right]^2  }{\left[\Phi\left( z_2  \right)  - \Phi\left( z_2-s_t  \right)\right] \left[\Phi\left( -z_2  \right)  + \Phi\left( z_2-s_t \right)\right]} dz_2 \nonumber \\
&= \int_{-\infty}^{0}\left(\frac{ \left[ \Phi'(z_2-s_t)-\Phi'(z_2)\right]^2  }{\Phi\left( z_2  \right)  - \Phi\left( z_2-s_t  \right) }+ \frac{ \left[ \Phi'(z_2-s_t)-\Phi'(z_2)\right]^2  }{\Phi\left( -z_2  \right)  + \Phi\left( z_2-s_t \right)}\right) dz_2
\label{upper.semi.interval}\end{align}
We show that both terms in \eqref{upper.semi.interval} are finite. 
For the first one we have
\begin{align*}
\int_{-\infty}^{0} \frac{ \left[ \Phi'(z_2-s_t)- \Phi'(z_2)\right]^2  }{\Phi\left( z_2  \right)  - \Phi\left( z_2-s_t  \right) } dz_2 &\leq \int_{-\infty}^{0} \frac{ \left[ \Phi'(z_2-s_t)-\Phi'(z_2)\right]^2  }{\Phi\left( z_2  \right)  - \Phi\left( z_2-s_0  \right) } dz_2\leq \int_{-\infty}^{0}\frac{ \left[\Phi'(z_2)\right]^2  }{\Phi\left( z_2  \right)  - \Phi\left( z_2-s_0  \right) } dz_2 \ ,
\end{align*}
and applying the same reasoning as before, we conclude that the integral is finite. Then for the second term we have
\begin{align*}
 \int_{-\infty}^{0}\frac{ \left[ \Phi'(z_2-s_t) -\Phi'(z_2)\right]^2  }{\Phi\left( -z_2  \right)  + \Phi\left( z_2-s_t \right)} dz_2 &\leq \int_{-\infty}^{0}\frac{ \left[ \Phi'(z_2-s_t)-\Phi'(z_2)\right]^2  }{\Phi\left( -z_2  \right)  } dz_2\leq 2\int_{-\infty}^{0} \left[\Phi'(z_2)\right]^2 dz_2 = \frac{1}{\sqrt[]{2}} \ .
\end{align*}

\textbf{Interval $(c_1,c_2)$:}
We proceed by applying a change of variable, and we arbitrarily choose to do it in the variable $z_2$. 
We get
\begin{align}
\int_{c_1}^{c_2}I(x,t)dx &= \int_0^{s_t} \left[ \frac{[\Phi'(z_2)-\Phi'(z_2-s_t)]^2}{\Phi(z_2) - \Phi(z_2-s_t) }+\frac{[\Phi'(z_2)-\Phi'(z_2-s_t)]^2}{\Phi(-z_2) + \Phi(z_2-s_t) }\right] dz_2
\label{middle.interval}\end{align}
and again we show that both integrals in \eqref{middle.interval} are bounded.
For the first integral we have
\begin{align*}
\int_0^{s_t} \frac{[\Phi'(z_2)-\Phi'(z_2-s_t)]^2}{\Phi(z_2) - \Phi(z_2-s_t) }dz_2&= 2\int_0^{s_t/2} \frac{[\Phi'(z_2)-\Phi'(z_2-s_t)]^2}{\Phi(z_2) - \Phi(z_2-s_t) }dz_2\\
&\leq 2\int_0^{s_t/2} \frac{[\Phi'(z_2)]^2}{\Phi(z_2) - \Phi(z_2-s_0) }dz_2\\
&\leq 2\int_0^{\infty} \frac{[\Phi'(z_2)]^2}{\Phi(z_2) - \Phi(z_2-s_0) }dz_2
\end{align*}
where the first equality holds because the function we are integrating is symmetric with respect $s_t/2$. 
The last integral is finite as it is trivially so on $[0,1]$ and using a comparison criteria with the function $f(z)=1/z^2$ 
it is also integrable on $[1,+\infty]$. 
In a similar way we analyze the second integral in \eqref{middle.interval} and by symmetry of the function with respect to the $s_t/2$ we get 
\begin{align*}
\int_0^{s_t}\frac{[\Phi'(z_2)-\Phi'(z_2-s_t)]^2}{\Phi(-z_2) + \Phi(z_2-s_t) } dz_2& =  2\int_0^{s_t/2} \frac{[\Phi'(z_2)-\Phi'(z_2-s_t)]^2}{\Phi(-z_2) + \Phi(z_2-s_t) }dz_2 \ .
\end{align*}
Then we compute the following bound
\begin{align*}
\Phi'(z_2)-\Phi'(z_2-s_t) &= \frac{1}{\sqrt[]{2\pi}}\left[ \exp\left(-\frac{z_2^2}{2}\right) -\exp\left(-\frac{(z_2-s_t)^2}{2}\right)\right] \\
&=\frac{1}{\sqrt[]{2\pi}}\exp\left(-\frac{z_2^2}{2}\right)\left[1-\exp\left(-\frac{s_t^2-2z_2s_t}{2}\right)\right]\\
&\leq \frac{1}{\sqrt[]{2\pi}}\exp\left(-\frac{z_2^2}{2}\right) \ ,
\end{align*}
where the  last inequality holds because $s_t^2-2z_2s_t\geq 0$ as $0\leq z_2\leq s_t/2$.
\begin{align*}
\int_0^{s_t}\frac{[\Phi'(z_2)-\Phi'(z_2-s_t)]^2}{\Phi(-z_2) + \Phi(z_2-s_t) } dz_2 &\leq \sqrt[]{\frac{2}{\pi}} \int_0^{s_t/2} \frac{\exp\left(-z_2^2\right)}{\Phi(-z_2) + \Phi(z_2-s_t) }dz_2 \leq  \int_0^{s_t/2} \frac{\exp\left(-z_2^2\right)}{\Phi(-z_2) }dz_2 \\
&\leq \int_0^1 \frac{\exp\left(-z_2^2\right)}{\Phi(-z_2) }dz_2 + \int_1^{+\infty} \frac{\exp\left(-z_2^2\right)}{\Phi(-z_2) }dz_2 \ .
\end{align*}
The first integral is trivially bounded. For the second to be bounded, we apply a comparison criteria with the function $f(z) = 1/z^2$. Putting together the given bounds we may bound the integral in \eqref{middle.interval} and the proof is finished.
\end{proof}

\begin{proof}[Proof of Proposition \ref{alpha.int.unbounded}]
We write 
\begin{equation}\label{prob.G.by.Gaussian.sum}
\PP(G= 1\ |\cF_t) =\sum_{k=-\infty}^{+\infty} \Phi\left(\frac{2k-B_t}{\sqrt{T-t}}\right) - \Phi\left(\frac{(2k-1)-B_t}{\sqrt{T-t}}\right) \ .
\end{equation}

Applying \eqref{ito.Gauss} to each term appearing in \eqref{prob.G.by.Gaussian.sum} and using  
the Dominated Convergence Theorem for semi-martingales \cite[Theorem IV.32]{Protter2005}
we get
\begin{equation} 
d\PP(G = 1 |\cF_t) =  \sum_{k=-\infty}^{+\infty} \frac{1}{\sqrt{T-t}} \left(\Phi'\left(\frac{2k-B_t}{\sqrt{T-t}}\right) - \Phi'\left(\frac{(2k-1)-B_t}{\sqrt{T-t}}\right) \right) dB_t \ .
\end{equation}
Then, the result follows on the same lines as the proof of Proposition \ref{alpha.int.dec}.
\end{proof}

\end{document}